%% file: ball_inheritance.tex
\documentclass{article}
\usepackage{amsmath,amssymb,amsfonts,amsthm, fullpage}

\newtheorem{lemma}{Lemma}
\newtheorem{definition}{Definition}
\newtheorem{observation}{Observation}
\newtheorem{theorem}{Theorem}

\newcommand{\dest}{\textrm{dest}}
\newcommand{\ball}{\textbf{b}}
\newcommand{\randi}{\textbf{i}}
\newcommand{\eps}{\varepsilon}

\newcommand{\Ins}{\mathcal{I}}
\newcommand{\Pairs}{\mathcal{Q}}
\newcommand{\Count}{\mathbf{X}}
\newcommand{\Rep}{\mathcal{R}}

\newcommand{\alphaval}{(S w \lg^{18} \lg n)/(n \lg n)}

\newcommand{\RepCells}{\mathcal{M}}
\newcommand{\E}{\mathbb{E}}
\renewcommand{\log}{\lg}

\title{Towards Tight Lower Bounds for Range Reporting on the RAM}

\author{Allan Gr\o nlund\thanks{Aarhus University. \texttt{jallan@cs.au.dk}. Supported by Center for Massive Data Algorithmics, a Center of the Danish National Research Foundation, grant DNRF84.} \and Kasper Green Larsen\thanks{Aarhus University. \texttt{larsen@cs.au.dk}. Supported by Center for Massive Data Algorithmics, a Center of the Danish National Research Foundation, grant DNRF84.}}

\begin{document}
\maketitle

\begin{abstract}
In the orthogonal range reporting problem, we are to preprocess a set
of $n$ points with integer coordinates on a $U \times U$ grid. The
goal is to support reporting all $k$ points inside an axis-aligned
query rectangle. This is one of the most fundamental data structure
problems in databases and computational geometry.  Despite the
importance of the problem its complexity remains unresolved in the
word-RAM.
On the upper bound side, three best tradeoffs exists:
\begin{enumerate}
\item Query time $O(\lg \lg n + k)$ with $O(n \lg^{\eps} n)$
  words of space for
  any constant $\eps>0$.
\item Query time $O((1 + k) \lg \lg n)$ with $O(n \lg \lg n)$ words of
  space.
\item Query time $O((1+k)\lg^{\eps} n)$ with optimal $O(n)$ words of space.
\end{enumerate}
However, the only known query time lower bound is $\Omega(\log \log n
+k)$, even for linear space data structures. 

All three current best upper bound tradeoffs are derived by reducing range
reporting to a \emph{ball-inheritance} problem. Ball-inheritance is a
problem that essentially encapsulates all previous attempts at solving
range reporting in the word-RAM. In this paper we make progress
towards closing the gap between the upper and lower bounds for range
reporting by proving cell probe lower bounds for ball-inheritance. Our
lower bounds are tight for a large range of parameters, excluding any
further progress for range reporting using the ball-inheritance
reduction.

\end{abstract}

\thispagestyle{empty} \setcounter{page}{0}
\newpage

\input{intro}

\input{lower_intro}
\input{lower}

\input{eliminating}
\input{scattering}
\input{shortcuts}

\bibliographystyle{abbrv}
\bibliography{bibliography}

\end{document}

%% file: intro.tex
\section{Introduction}
In the orthogonal range reporting problem, we are to preprocess a set
of $n$ points with integer coordinates on a $U \times U$ grid. The
goal is to support reporting all $k$ points inside an axis-aligned
query rectangle. This is one of the most fundamental data
structure problems in databases and computational geometry. Given the
importance of the problem, it has been extensively studied in all the
relevant models of computation including e.g. the word-RAM, pointer
machine and external memory model. In the latter
two models, we typically work under an assumption of
\emph{indivisibility}, meaning that input points have to be stored as
they are, i.e. compression techniques such as rank-space reduction and word-packing
cannot be used to reduce the space consumption of data structures. The
indivisibility assumption greatly alleviates the task of proving lower
bounds, which has resulted in a completely tight characterisation of
the complexity of orthogonal range reporting in these two
models. More specifically, Chazelle~\cite{Chazelle.filtering.search} presented a pointer
machine data structure answering queries in optimal $O(\lg n+k)$ time
using $O(n \lg n/ \lg \lg n)$ space and later proved that this space
bound is optimal for any query time of the form $O(\lg^c n+k)$, where
$c \geq 1$ is an arbitrary constant~\cite{Chazelle.LB.reporting}. In the external memory
model, Arge et al.~\cite{arge:indexandrange} presented a data structure answering
queries in optimal $O(\lg_B n + k/B)$ I/Os with $O(n \lg n/ \lg \lg_B
n)$ space and also proved the space bound to be optimal for any query
time of the form $O(\lg^c_B n + k/B)$, where $c \geq 1$ is a
constant. Here $B$ is the disk block size. Thus the orthogonal range
reporting problem has been completely closed for at least 15 years in
both these models of computation. If we instead abandon the
indivisibility assumption and
consider orthogonal range reporting in the arguably more
realistic model of computation, the word-RAM, our understanding of the problem is much
more disappointing. Assuming the coordinates are polynomial in $n$ ($U = n^{O(1)}$), the current best word-RAM data structures,  by Chan et al.~\cite{patrascu11range}, achieve
the following tradeoffs:
\begin{enumerate}
\item Optimal query time $O(\lg \lg n + k)$ with $O(n \lg^{\eps} n)$
  words of space for
  any constant $\eps>0$.
\item Query time $O((1 + k) \lg \lg n)$ with $O(n \lg \lg n)$ words of
  space.
\item Query time $O((1+k)\lg^{\eps} n)$ with optimal $O(n)$ words of space.
\end{enumerate}
Thus we can achieve linear space by paying a $\lg^\eps n$ penalty per
point reported. And even if we insist on an optimal $O(\lg \lg n +
k)$ query time, it
is possible to improve over the optimal space bound in the pointer machine
and external memory model by almost a $\lg n$ factor. Naturally the
improvements rely heavily on points not being indivisible.

On the lower bounds side, P{\v a}tra{\c s}cu and
Thorup~\cite{patrascu06pred,patrascu07randpred} proved that
the query time must be $\Omega(\lg \lg n+k)$ for space $n \lg^{O(1)}
n$. This lower bound was obtained by reduction from
the predecessor search problem. For predecessor search, the query time
of $\lg \lg n$ is
known to be achievable with linear space. Thus the
reduction is incapable of distinguishing the three space regimes
of the current best data structures for range reporting. Perhaps it
might just be possible to construct a linear space data structure with
$O(\lg \lg n+ k)$ query time. This would have a huge impact in
practice, since the non-linear space solutions are most often
abandoned for the
kd-trees~\cite{bentley75:kd_tree}, using linear space and answering
queries in $O(\sqrt{n}+k)$ time. This is simply
because more than a constant factor above linear space is prohibitive
for most applications. Thus ruling out the existence of fast linear space
data structures would be a major contribution. The focus of this paper
is on understanding this gap and
the complexity of orthogonal range reporting in the word-RAM. This boils down to understanding how much
compression and word-packing techniques can help us in the regime between
linear space and $O(n \lg^{\eps} n)$ space. Since our
results concern definitions made by Chan et al.~\cite{patrascu11range}, we first give a more formal definition of the word-RAM and
briefly review the technique of rank space reduction and the main ideas in~\cite{patrascu11range}. 

\subsection{Range Reporting in the word-RAM}
The word-RAM model was designed to mimic what is possible in modern
imperative programming languages such as C. In the word-RAM, the
memory is divided into words of $\Theta(\lg n)$ bits. The words have
integer addresses and we allow random access to any word in constant
time. We also assume all standard word operations from modern
programming languages takes constant time. This includes
e.g. integer addition, subtraction, multiplication, division, bit-wise
AND, OR, XOR, SHIFT etc. Having $\Theta(\log n)$ bit words is a
reasonable assumption since machine words on standard computers have enough bits to
address the input and to store pointers into a data
structure.

\paragraph{Rank Space Reduction.}
Most of the previous range reporting data structures for the word-RAM
have used rank space reduction (or variants thereof) to save space, see
e.g.~\cite{Brodal00h, Nekrich:linear}. Rank space reduction is
the following: Given a set $P$ of $n$ points on a $U \times U$ grid,
compute for each point $(x,y) \in P$ the rank $r_x(x)$ of $x$ amongst the
$x$-coordinates of points in $P$. Similarly compute the rank $r_y(x)$ of
$y$ amongst the $y$-coordinates of points in $P$. Construct a new
point set $P^*$ with each point $(x,y) \in P$ replaced by
$(r_x(x),r_y(y))$. The point set $P^*$ is said to be in \emph{rank
  space}. A point $(x,y) \in P$ lies inside a query range $q = [x_0 ; x_1]
\times [y_0 ; y_1]$ precisely if $(r_x(x),r_y(y))$ lies inside the
range $q^* = [r_x(x_0) ; r_x(x_1)] \times [r_y(y_0) ;
r_y(y_1)]$. Thus if we store a data structure for mapping $q$ to $q^*$
and a table mapping points in $P^*$ back to
points in $P$, the output of a query $q$ can be computed from the
output of the
query $q^*$ on $P^*$. Since the coordinates of a point in $P^*$ can be
represented using $\lg n$ bits, this gives a saving in space if $\lg n
\ll \lg U$. 

In previous range reporting data structures, rank space reductions are
often used recursively on smaller and smaller point sets $P_t \subset
P_{t-1} \subset \cdots \subset P_1 \subset P$. Applying $t$ rounds of
rank space reduction however results in a query time
of $O(f(n) + t k)$ since each reported point has to be
\emph{decompressed} through $t$ rank space reduction tables.

\paragraph{The Ball-Inheritance Problem.}
In the following, we present the main ideas of the current best data
structures, due to Chan et al.~\cite{patrascu11range}. Their solution is based on an
elegant way of combining rank space reductions over all levels of a
range tree:

Construct a complete binary tree with the
points of $P$ stored in the leaves ordered from left to right by their
$x$-coordinate. Every internal node $v$ is associated with the subset
of points $P_v$ stored in the leaves of the subtree rooted at $v$. For
every internal node $v$, map the points $P_v$ to rank space and denote
the resulting set of points $P^*_v$. Store in $v$ a data structure for
answering 3-sided range queries on $P^*_v$. Here a 3-sided query is
either of the form $[x_0 ; \infty) \times [y_0 ; y_1]$ or $(-\infty,
x_1] \times [y_0 ; y_1]$. If we require that only the rank space
$y$-coordinate of a point is reported (and not the rank space
$x$-coordinate), these 3-sided data structures can be
implemented in $O(n)$ bits and with $O(k)$ query time using succinct data
structures for range minimum queries, see e.g.~\cite{fischer10rmq}. For each leaf,
we simply store the associated point. The total space usage is $O(n
\lg n + n\lg U)$ bits, which is $O(n)$ words.

To answer a query $q = [x_0 ;
x_1] \times [y_0 ; y_1]$, find the lowest common ancestor, $w$, of the
leaves storing the successor of $x_0$ and the predecessor of $x_1$
respectively. Let $w_\ell$ be the left child of $w$ and $w_r$ the
right child. The points inside $q$ are precisely the points
$P_{w_\ell} \cap [x_0 ; \infty) \times [y_0 ; y_1]$ plus 
$P_{w_r} \cap (-\infty, x_1] \times [y_0 ; y_1]$. The data structures of Chan
et al. now proceeds by mapping these two 3-sided queries to rank space
amongst points in $P_{w_\ell}^*$ and $P_{w_r}^*$ respectively and
answering the two queries using the 3-sided data structures stored at
$w_\ell$ and $w_r$. This
reports, for every point $(x,y) \in P_{w_\ell} \cap q$ (and $(x,y) \in
P_{w_r} \cap q$), the rank of $y$ amongst the $y$-coordinates of all
points in $P_{w_\ell}$ ($P_{w_r}$). Assuming one can build an $S$ word
auxiliary data structure that supports mapping these rank space
$y$-coordinates back to the original points in $t$ time per point (i.e. through
$t$ rank space decompressions), this gives a data structure for orthogonal
range reporting that answers queries in $O(\lg \lg n + t(1+k))$ time
using $S+O(n)$ space, see~\cite{patrascu11range} for full details. Chan et al. named this abstract decompression
problem \emph{the ball-inheritance problem} and defined it as
follows:

\begin{definition}[Chan et al.~\cite{patrascu11range}]
In the ball-inheritance problem, the input is a complete binary tree
with $n$ leaves. In the root node, there is an ordered list of $n$
balls. Each ball is associated with a unique leaf of the
tree and we say the ball \emph{reaches} that leaf. Every internal node $v$ also has an associated list of balls,
containing those balls reaching a leaf in the subtree
rooted at $v$. The ordering of the balls in $v$'s list is the same as
their ordering in the root's list. We think of each ball in $v$'s list
as being inherited from $v$'s parent.

A query is specified by a pair $(v,i)$ where $v$ is a node in the tree
and $i$ is an index into $v$'s
list of balls. The goal is to return the index of the leaf reached by
the $i$'th ball in $v$'s list of balls.
\end{definition}

It is not hard to see that a solution to the ball-inheritance problem
is precisely what is needed in Chan et al.'s data structures: We have
one ball per point. The ball corresponding to a point $(x,y)$ reaches
the $r_x(x)$'th leaf, where $r_x(x)$ is the rank of $x$ amongst all
$x$-coordinates. The ordering of the balls inside the lists is just the ordering on
the $y$-coordinates of the corresponding points. Thus answering a
ball-inheritance query $(v,i)$ corresponds exactly to
determining the leaf storing the point from $P_v$ having a rank space
$y$-coordinate of $i$. Since Chan et al. stored
the points in the leaves, this also recovers the original point.

All three tradeoffs by Chan et al. come from solving the
ball-inheritance problem with the following bounds:

\begin{theorem}[Chan et al.~\cite{patrascu11range}]
\label{thm:ballupper}
For any $2 \leq B \leq \lg^\eps n$, we can solve the ball-inheritance
problem with: (1) space $O(nB\lg \lg n)$ and query time $O(\lg_B \lg
n)$; or (2) space $O(n \lg_B \lg n)$ and query time $O(B \lg \lg n)$.
\end{theorem}

While not all previous range reporting data structures directly solve
the ball-inheritance problem, they are all based on rank space
reductions and decompression of one point at a time, just in less
efficient ways. Thus the
ball-inheritance problem in some sense captures the essence of all
previous approaches to solving range
reporting and the bounds obtained for the ball-inheritance problem
also sets the current limits for orthogonal range reporting.

We remark
that the ball-inheritance problem also has been used to improve the
upper bounds for various other
problems with a range reporting flavour to them, see
e.g.~\cite{chan:adaptive, brodal:skyline}. Thus
in light of the lack of progress in proving tight lower bounds for range
reporting, it seems like a natural goal to understand the complexity
of the ball-inheritance problem.

\subsection{Our Results}
In this paper, we prove a lower bound for the ball-inheritance
problem. Our lower bound is tight for a large range of parameters and
is as follows:

\begin{theorem}
\label{thm:rammain}
Any word-RAM data structure for the ball-inheritance problem which uses $S$
words of space, must have query time $t$ satisfying:
$$
t = \Omega\left(\frac{\lg \lg n}{\lg(S/n) + \lg \lg \lg n} \right)
$$
\end{theorem}

Comparing to the ball-inheritance upper bounds of Chan et al. (Theorem~\ref{thm:ballupper}), we see that this essentially matches their
first tradeoff and is tight for any $S = \Omega(n \lg^{1+\eps} \lg n)$
where $\eps>0$ is an arbitrarily small constant. Most importantly, it
implies that for constant query time, one needs space $n \lg^{\eps}n$
words. Thus any range reporting data structure based on the
ball-inheritance problem cannot improve over the bounds of Chan et
al. in the regime of space $S = \Omega(n \lg^{1 + \eps} \lg n)$
words. We believe this holds true for any data structure that is based
on \emph{decompressing} one point at a time from some subproblem in
rank space. Since decompressing from a subproblem in rank space is hard to formalize exactly, we leave
it at this.

One can view our lower bound in two ways: Either as a strong indicator
that the data structure of Chan et al. is optimal, or as a suggestion
for how to find better upper bounds. The lower bound above shows that
if we want to develop faster data structures, we have to find a
technique that in some sense allows us to decompress $\omega(1)$
points in one batch, faster than decompressing each point in turn. This is not necessarily impossible given the
large success of batched evaluations in other problems such as matrix
multiplication and multipoint evaluation of polynomials.

We also want to make a remark regarding the gap between the second
tradeoff of Chan et al. and our lower bound. We conjecture that the
upper bound of Chan et al. is tight, but note that current lower bound
techniques (in the cell probe model) are incapable of proving any lower
bounds exceeding the one we obtain in Theorem~\ref{thm:rammain}:
The ball-inheritance problem has only $n \lg n$ queries and the
strongest lower bound for any data structure problem with $m$ queries
(for any $m$) is $t = \Omega(\lg(m/n)/\lg(S/n))$~\cite{larsen:staticloga}, thus apart from our
$\lg \lg \lg n$ ``dirt factor'', our lower bound is as strong as it
possibly can be with current techniques.

\paragraph{Technical Contributions.}
As a side remark, we believe our lower bound proof has interest from a
purely technical point of view. In the lower bound proof, we carefully
exploit that a data structure is \emph{not}
non-deterministic. While this might sound odd at first, Wang and
Yin~\cite{wang:nondeter} recently showed that, with only few
exceptions (e.g. the predecessor
lower bounds), all previous
lower bound techniques yield lower bounds that
hold non-deterministically. Thus having a new proof outside this
category is an important contribution and may hopefully help in
closing fundamental problems where avoiding non-determinism in proofs is
crucial. This is e.g. the case for the deterministic dictionaries problem, which is
amongst the most fundamental open problems in the field of data
structures. This problem is trivially solved with constant update time
and query time non-deterministically
(just maintain a sorted linked list) and hence lower bound proofs must use
ideas similar to those we present in this paper to prove super
constant lower bounds for this important problem.

%% file: lower_intro.tex
\section{Lower Bound Proof}
\label{sec:lower_bound_proof}
We prove our lower bound in the cell probe model~\cite{yao:cellprobe},
where the complexity of a data structure is the number of cells it
reads/probes. More specifically, a data structure with query time $t$ and
space $S$ consists of memory of $S$
cells with consecutive integer addresses $0,\dots,S-1$. Each cell
stores $w$ bits and we assume $w=\Omega(\log
n)$. When answering a query, the data structure may probe up to
$t$ cells and must announce the answer to the query solely based on
the contents of the probed cells. The cell to probe in each step may
depend arbitrarily on the query and the contents of previously probed
cells. Thus computation is free of charge in the cell probe model and
lower bounds proved in this model clearly applies to word-RAM data
structures.

\subsection{Main Ideas}
\label{sec:ideas}
In the following, we sketch the overall approach in our proof. Assume
we have a data structure for the ball-inheritance problem, having
space $S$ cells of $w$ bits and with query time $t$. Assume
furthermore that the data structure performs very poorly in the
following sense: For every input $I$ to the ball-inheritance
problem and every leaf index $b \in [n] = \{0,\dots,n-1\}$, let $Q(b, I)$ be the set
of queries that have $b$ as its answer. Each such query probes at most
$t$ cells of the data structure on input $I$. Assume these sets of
cells are \emph{disjoint}, i.e. information about the leaf $b$ is
stored in $|Q(b,I)| = \lg n$ disjoint $t$-sized locations in the memory. 

Now pick a uniform random set $C$ of $\lg(n!)/(4w)$ memory
cells. For a query $q$, we say that $q$ survives if all its $t$
probes lie in $C$. Then by the disjointness of the probed cells, there
will be a surviving query in $Q(b,I)$ with probability roughly
$1-(1-(|C|/S)^t)^{\lg n}$. If $t = o(\lg \lg n/\lg(S/|C|))$, this
is about $1-\exp(\lg n \cdot (|C|/S)^t) = 1-\exp(\lg^{1-o(1)} n)$,
i.e. each leaf index is almost certainly the answer to a surviving
query. Thus $C$ must basically store the entire input. But $|C|$ is
too small for this and we get a contradiction, i.e. $t = \Omega(\lg
\lg n/\lg(Sw/(n \lg n)))$, which roughly equals the lower bound we claim. There are
obviously a few more details to it, but this is the main idea.

Of course any realistic attempt at designing a data structure for the
ball-inheritance problem would try to make the queries in $Q(b,I)$
probe the same cells (which is exactly what Chan et al.'s solution does~\cite{patrascu11range}). In our actual proof, we get around this using
the following observation: Consider two queries $q_1,q_2$ to the
ball-inheritance problem, where $q_2$ is asked in a node $d$ levels below
the node of $q_1$. The probability $q_1$ and $q_2$ return the same
leaf index is exponentially decreasing in $d$. In particular this
means that for the very first probe, the queries in $Q(b,I)$ will
almost certainly read different cells, which is precisely the property
we exploited above. If we pick a random sample of cells, there
will be many queries in $Q(b,I)$ that have their first probe
in the sample. To handle the remaining $t-1$ probes, we follow~\cite{patrascu06pred} and extend the cell probe model with the
concepts \emph{published bits} and \emph{accepted queries}. A data
structure is allowed to publish bits at preprocessing time that the
query algorithm may read free of charge. After inspecting a given
query and the published bits, a data structure can choose to reject the
query and not return an answer. Otherwise, the query is accepted
and the algorithm must output the correct answer. Note that it is only
allowed to reject queries before performing any probes.

The crucial idea is now the following: If the data structure has few
published bits, then for most leaves $b \in [n]$, the published bits
simply contains too little information to make the queries in $Q(b,I)$ probe the same
cells. Thus for $t$ rounds, we can pick a random sample of cells and
\emph{publish} their contents. For every accepted query, we check if
its first probe is amongst the published cells. If so, we continue to accept it and
may skip the first probe since we know the contents of the requested
cell. Otherwise we simply reject it. If the published cell sets are small enough,
there continues to be too little information in the published bits to
make the queries in $Q(b,I)$ meet. Since this holds for all $t$
probes, the argument above for the poorly performing data structures
carry through and we get our lower bound.

%% file: lower.tex
\subsection{Deriving the Lower Bound}
With the ideas from Section~\ref{sec:ideas} in mind, we present our
technical lemma that allows us to publish bits for $t$ rounds to
eliminate probes while ensuring that most leaves are still the answer
to many accepted queries. Before we
present the lemma, consider partitioning the ball-inheritance tree into into $\log n / Y$
disjoint layers of $Y$ consecutive tree levels and group the
accepted queries by these layers. Think of $Y$ as looking at the
queries at a given zoom level. To measure how much information we have
left about the different leaves, we count for each leaf $b \in [n]$ how many
layers that have at least one accepted query with $b$ as its answer. If this
count is large, then intuitively the answers to all accepted queries
carry much information.

Formally, given a data structure for the ball-inheritance problem,
define for every $1 \leq Y \leq \lg n$ and index $i \in [\lg n/Y]$ the
\emph{query-support set} of a leaf $b \in [n]$ on an input $I$ as the
set $Q^Y_i(b, I)$ of accepted queries in the tree levels $\{iY, \dots,
(i+1)Y-1\}$ that has $b$ as its answer.
Observe that $|Q^Y_i(b, I)| \in \{0,\dots, Y\}$ since there is
precisely one query in each tree level that has $b$ as its
answer (it may be less than $Y$ since some queries might be rejected). Define also the $Y$-\emph{level-support} of an input $I$,
denoted $L^Y(I)$, as the the number of pairs $(b,i)$ such that
$Q^Y_i(b,I)$ is non-empty.

With this notation in hand we are ready to state our main Probe Elimination Lemma.
\begin{lemma}
\label{lem:elim}
Let $\Ins$ be a set of inputs to the ball-inheritance problem where
$|\Ins| \geq n!/2^n$. Assume a ball-inheritance data structure uses
$S$ cells of $w$ bits, answers queries in $t$ probes, has $p < n\lg n/\lg^9 \lg n$ published bits and
satisfies $L^Y(I) \geq (1-1/Z)n \lg n/Y$ for all $I \in \Ins$ for some
parameters $Z \geq 2$ and $64 \lg w \leq Y \leq \lg n/\alpha$, where
$\alpha=\alphaval$. Then there exists a subset of inputs $\Ins^* \subseteq
\Ins$, with $|\Ins^*| \geq |\Ins|/2$, and another ball-inheritance data structure using $S$
cells of $w$ bits, answering queries in $t-1$ probes with $p+O(n \lg
n/ \lg^{10} \lg n)$ published bits, and satisfying $L^{\alpha
  Y}(I) \geq (1-1/Z-1/\lg \lg^3 n)n \lg n/(\alpha Y)$ for all $I \in
\Ins^*$.
\end{lemma}

In laymans terms, the lemma states that we can decrease the number of
probes of a data structure by one, while
only increasing the published bits with a lower order term. When we do
this, we maintain the essential property that the leaves
still have high support, just on a coarser zoom level.
The $Z$ factor is basically just a \emph{dirt} factor. We defer the
proof of
Lemma~\ref{lem:elim} to Section~\ref{sec:elim}. In the
following we instead use Lemma~\ref{lem:elim} to prove our main result, Theorem~\ref{thm:rammain}.

Assume for contradiction that a ball-inheritance data
structure exists satisfying $t = o(\lg \lg_w n /\lg \alpha)$, where
$\alpha = \alphaval$. We proceed by repeatedly applying
Lemma~\ref{lem:elim} to eliminate all $t$ probes of the data
structure. In order to guarantee we can apply Lemma~\ref{lem:elim} $t$
times, we check the conditions for applying it. The conditions involve
the number of published bits $p$, the parameters $Z$ and $Y$ and
$|\Ins|$. The values of these parameters will change for each
application, thus we use $p^{(i)}, Z^{(i)},Y^{(i)}$ and $|\Ins^{(i)}|$
to denote these parameters just before the $i$'th invocation of the
lemma. For the first round, we have $p^{(1)} = 0$ and
$|\Ins^{(1)}|=n!$. Note also that $L^Y(I) = n \lg n/Y$ for any $Y$
before the first round. Thus we choose $Y^{(1)} = 64 \lg w$ to satisfy
the conditions $64 \lg w \leq Y^{(1)} \leq \lg n/\alpha$. This also
means that we are free to choose $Z^{(1)} \geq 2$ as we wish. We
simply let $Z^{(1)} = \lg^3 \lg n$. 
%
Examining the lemma, we conclude that our parameters evolve in the
following way (assuming we do not violate any of the conditions):
$$
p^{(1+i)}=O(i(n \lg n / \lg^{10} \lg n)),\quad
|\Ins^{(1+i)}| \geq n!/2^i,\quad
Y^{(1+i)} = 64 \lg w \cdot \alpha^i,\quad
Z^{(i)} \geq  \lg^{3}\lg n/ i.
$$
Since we assumed $t=o(\lg \lg_w n/\lg \alpha)$, this means that
$$
p^{(1+t)}=o(n \lg n / \lg^{9} \lg n),\quad
  |\Ins^{(1+t)}| \geq n!/\lg n, \quad
  Y^{(1+t)} =
  o(\lg n), \quad
Z^{(1+t)} \geq  \lg^{2}\lg n.
$$
We conclude that we can apply our lemma for $t$ rounds under the
contradictory assumption. Furthermore, the data structure we are left
with answers queries in $0$ probes on a subset $\Ins^* = \Ins^{(1+t)}$
of inputs, where $|\Ins^*| \geq n!/\lg n$. It has $o(n \lg n/\lg^9 \lg
n)$ published bits and there is some $Y^* = o(\lg n)$ such that
$L^{Y^*}(I) \geq (1-1/\lg^2\lg n)n \lg n/Y^*$ for all $I \in
\Ins^*$. That this is contradictory should not come as a surprise: our
$0$-probe data structure is capable of answering queries about almost
all leaves using only the $o(n \lg n/\lg^{9} \lg n) \ll \lg |\Ins^*|$
published bits. The formal argument we use to reach the contradiction
is as follows: we show that the $0$-probe data structure can be used
to uniquely encode every input $I \in \Ins^*$ into a bit string of
length less than $\lg(|\Ins^*|) = \lg(n!)-\lg \lg n$ bits. This gives
the contradiction since there are fewer such bit strings than
inputs. We present the encoding and decoding algorithms in the following:

\paragraph{Encoding.}
Let $I \in \Ins^*$ be an input to encode. Observe that if we manage to encode
the leaf index reached by each ball in the root node's list of balls,
then that information completely specifies $I$. With this in mind, we implement the $0$-probe
data structure above on $I$ and proceed as follows:
\begin{enumerate}
\item First we write down the published bits on input $I$. This cost
  $o(n \lg n/\lg^9 \lg n)$ bits.
\item For $i=1,\dots,n$ consider the $i$'th ball in the root node's
  list of balls. Let $b_i$ denote the index of the leaf reached by
  that ball. We write down $\lg n/2$ bits for each such ball in
  turn, specifying the subtree at depth $\lg n/2$ that contains the
  leaf $b_i$. This costs $n \lg n/2$ bits.
\item Finally, we go through all leaf nodes from left to right. For a
  leaf $b$, we check if there is an accepted query returning $b$ as
  its answer amongst all queries in all nodes of depth at most $\lg
  n/2$. If so, we continue to the next leaf. Otherwise we write $\lg
  n$ bits denoting the rank of the ball reaching $b$ amongst balls the root node's list of
  balls. If $X$ is the number of leaves with no accepted query
  reporting it in tree levels $\{0,\dots,\lg n/2\}$, this step costs
  $X \lg n$ bits.
\end{enumerate}

\paragraph{Decoding.}
To recover $I$ from the above encoding, we do as follows.
\begin{enumerate}
\item We first go through all nodes $v$ of depth $d$ for
  $d=0,\dots,\lg n/2$. For each such node, let $q^v_1,\dots,q^v_{n/2^d}$
  denote the queries we can ask at node $v$, i.e. $q^v_i$ asks for the
  leaf reached by the $i$'th ball in $v$'s list of balls. We run the
  query algorithm for each
  such query in turn using the published bits written in step 1. of
  the encoding procedure. Since our data structure makes $0$ probes, this returns the answer to each such
  accepted query, i.e. we have collected a set $\Pairs$ of pairs $(q^v_i, b)$
  such that $b$ is the index of the leaf reached by the $i$'th ball in
  $v$'s list of balls. 
\item We now partition $\Pairs$ into one set $\Pairs_b$ for each leaf
  index $b$. The set $\Pairs_b$ contains all pairs $(q^v_i, b') \in
  \Pairs$ such that $b'=b$. There are precisely $X$ empty such sets.
\item For each empty set $\Pairs_b$ in turn (ordered based on $b$), we
  use the bits written in step 3. of the encoding procedure to recover
  the rank of the ball reaching $b$ amongst all balls in the root node's list of balls.
\item For every non-empty set $\Pairs_b$, pick an arbitrary pair
  $(q^v_i, b) \in \Pairs_b$. From this pair alone, we know that the
  ball reaching $b$ has rank $i$ amongst all balls ending in a leaf of
  the subtree rooted at $v$. Now initialize a counter $\Delta$ to
  $0$. Using the bits written in step 2. of the encoding procedure, we
  now go through all balls in the root node's list of balls in
  turn. For the $r$'th ball, $r=1,\dots,n$, we check the $\lg n/2$
  bits written for it and from this we determine if the ball reaches a
  leaf in $v$'s subtree (possible since $v$ can only be in the first
  $\log n/2$ levels by construction). If so, we increment $\Delta$ by
  $1$. If this causes $\Delta$ to reach $i$, we conclude that the ball
  ending in $b$ has rank $r$ in the root node's list of balls.
\item From the above steps, we have for every leaf $b$ determined the
  rank of the ball reaching it amongst all balls in the root node's
  list of balls. This information completely specifies $I$.
\end{enumerate}

\paragraph{Analysis.}
The encoding above costs
$$
o(n \lg n/\lg^9 \lg n) + n \lg n/2 + X \lg n
$$
bits. Now observe that if $\Pairs_b$ is empty for a leaf index $b$,
this means $Q_i^{Y^*}(b,I)$ is empty for every $i \in \{0,\dots, \lg
n/(2Y^*)-1\}$. This gives $L^{Y^*}(I) \leq n \lg n/Y^* - X(\lg
n/(2Y^*))$. But we know $L^{Y^*}(I) \geq (1-1/\lg^2 \lg n)n\lg n/Y^*$
and we conclude
$$
X \leq 2n/\lg^2\lg n.
$$
The encoding thus costs
$$n\lg n/2 + O(n \lg n/\lg^2 \lg n).$$
Since $\lg(n!) = n \lg n - O(n)$, we conclude that our encoding uses
no more than
$$
\lg(|\Ins^*|)-n\lg n/2 + O(n \lg n/\lg^2 \lg n) =
\lg(|\Ins^*|)-\Omega(n \lg n)
$$
bits, which completes the proof. 

We have thus shown $t = \Omega(\lg
\lg_w n / \lg \alpha)$ where $\alpha = \alphaval$. In the word-RAM, we
assume $w = \Theta(\lg n)$ and the lower bound becomes the claimed $t
= \Omega(\lg \lg n/(\lg(S/n) + \lg \lg \lg n))$.

%% file: eliminating.tex
\subsection{Eliminating Probes}
\label{sec:elim}
In this section we prove Lemma~\ref{lem:elim}. Recalling the intuition
presented in Section~\ref{sec:ideas}, we want to show that for a data
structure with few published bits, the different accepted queries reporting a
fixed leaf index $b \in [n]$ have to probe distinct cells in their
first probe. If we can
establish this, we can pick a small random sample of memory cells and
there are many of the accepted queries that make their first probe in the sample.

To formalize the above, we define a memory cell $c$ to be $k$-popular
on input $I$, if at least $k$ accepted queries make their first probe
in $c$ on $I$. Define for every query-support set $Q_i^Y(b, I)$ the
\emph{cell-support set} $C_i^Y(b, I)$ as the set of memory cells that
are read in the first probe of a query in $Q_i^Y(b, I)$ on input
$I$. We measure to what extend the queries in $Q_i^Y(b,I)$ probe
distinct cells using the following definitions.
\begin{definition}
\label{def:scatter}
For an input $I$ and value $Y \in \{1,\dots,\lg n\}$, we say that a
pair $(b, i)$, where $b \in [n]$ and $i \in \{0, \dots, \lg n/Y -1\}$,
is $Y$-\emph{scattered} on input $I$ if one of the following three
holds:
\begin{enumerate}
\item $Q_i^{Y}(b, I)$ contains a query making $0$ probes.
\item $C_i^{Y}(b, I)$ contains a $w^3$-popular cell.
\item $|C_i^{Y}(b, I)| \geq \alpha/\lg^6\lg n$.
\end{enumerate}
We define the $Y$-\emph{scatter-number} of $I$, denoted $\Gamma^Y(I)$,
as the number of pairs $(b,i)$ that are $Y$-scattered on $I$.
\end{definition}
If a query makes zero probes, all the information needed to answer
it is contained in the already published bits.  There are very few
$w^3$-popular cells, so publishing all of them costs few bits.  Most
interestingly, if the queries in each support $Q_i^Y(b,I)$ set probe
many distinct cells in their first probe (case 3.), then a random sample of
cells will contain at least one of these cells with good probability.

We need the following lemma that captures the correspondence between
large support on zoom level $Y$, the properties maintained by our Probe
Elimination Lemma, and large scattering on a higher zoom level $\alpha
Y$.
\begin{lemma}
\label{lem:scatter}
Let $\Ins$ be a set of inputs to the ball-inheritance problem where $|\Ins| \geq n!/2^n$. Assume
a ball-inheritance data structure uses $S$ cells of $w$ bits, has $p <
n \lg n/\lg^9 \lg n$ published bits and satisfies $L^Y(I) \geq (1-1/Z)n \lg
n/Y$ for all $I \in \Ins$ for some parameters $Z\geq 2$ and $64 \lg w \leq Y \leq
\lg n/\alpha$, where $\alpha=\alphaval$. Then there exists a subset $\Ins^*
\subseteq \Ins$ of inputs such that $|\Ins^*| \geq |\Ins|/2$ and
$$\Gamma^{\alpha Y}(I) \geq \left(1-\frac{1}{\lg^3\lg n}\right)\cdot \left(1-\frac{1}{Z}\right) \cdot \frac{n\lg n}{\alpha Y}.$$
for all $I \in \Ins^*$.
\end{lemma}

We defer the proof of Lemma~\ref{lem:scatter} to
Section~\ref{sec:highscatter}, and use it to prove Lemma~\ref{lem:elim} instead.
Let $\Ins$ be a set of at least $n!/2^n$ inputs to the ball
inheritance problem. Assume furthermore we are given a ball
inheritance data structure that uses $S$ cells of $w$ bits, answers
queries in $t$ probes, has $p < n \lg n/\lg^9 \lg n$ published bits,
and satisfies $L^Y(I) \geq (1-1/Z)n \lg n/Y$ for all $I \in
\Ins$ for some parameters $Z \geq 2$ and $64 \lg w \leq Y \leq \lg
n/\alpha$ where $\alpha=\alphaval$ (as in the assumptions of
Lemma~\ref{lem:elim} and Lemma~\ref{lem:scatter}). Let $\Ins^*
\subseteq \Ins$ be the subset of $\Ins$ promised by
Lemma~\ref{lem:scatter}. Our goal is to construct a new ball
inheritance data structure answering queries in $t-1$ probes for the
inputs $\Ins^*$ while publishing few bits and keeping $L^{\alpha
  Y}(I)$ fairly large for all $I \in \Ins^*$. Given an input $I \in
\Ins^*$, we keep the (old) $p$ published bits and publish some
additional bits from our data structure as follows:
\begin{enumerate}
\item First we publish all memory cells that are $w^3$-popular on
  input $I$. Since
  there are no more than $n \lg n$ accepted queries, there are no more than $n \lg n /w^3$
  popular cells. The addresses and contents of all such cells can be
  described using $O(n \lg n /w^2) = O(n/ \lg n)$
  bits.
\item Next we collect all $\alpha Y$-scattered pairs $(b,i)$ for input
  $I$. We remove those pairs for which $Q_i^{\alpha Y}(b,I)$ contains a query
  making $0$ probes, or $C_i^{\alpha Y}(b,I)$ contains a $w^3$-popular
  cell. By definition, the remaining $\alpha Y$-scattered pairs $(b,i)$ must satisfy
  $|C_i^{\alpha Y}(b,I)| \geq \alpha/\lg^6 \lg n$. We now consider all subsets
  of $n
  \lg n / (w \lg^{10} \lg n)$ memory cells and
  publish the subset $P^* \subseteq [S]$ for which most remaining pairs $(b,i)$
  satisfies $C_i^{\alpha Y}(b,I) \cap P^* \neq \emptyset$. Specifying the
  addresses and contents of cells in $P^*$ costs $O(n \lg n/\lg^{10}
  \lg n)$ bits.
\end{enumerate}

The query algorithm of our modified data structure is simple: We start
running the old query algorithm with the $p$ ``old'' published
bits and stop once one of the following happens:
\begin{enumerate}
\item If the old query algorithm rejects the query, we also reject it.
\item If the old query algorithm answers the query without any probes,
  we know the answer to the query and return it.
\item Otherwise the old query algorithm makes at least one memory
  probe. The (address of the) first cell probed, denoted $c$, can be determined solely from the old
  published bits. Before making the actual probe, we check the newly
  published cells to see if $c$ is amongst them. If so, we have the
  contents of $c$ in the published bits and therefore skip the
  probe. We then continue executing the old query algorithm and
  have successfully reduced the number of probes by one. If $c$ was
  not published, we simply reject the query.
\end{enumerate}

Clearly our new data structure answers queries in $t-1$ probes and has
$p+O(n \lg n/\lg^{10} \lg n)$ published bits. What remains is to argue
that $L^{\alpha Y}(I)$ is high for all $I \in \Ins^*$ for this new
data structure. To distinguish the new data structure and the old, we
use $\bar{L}, \bar{Q}$ and $\bar{\Gamma}$ in place of $L,Q$ and
$\Gamma$ when referring to
the new data structure. $L, Q$ and $\Gamma$ refers to the old data structure.

So fix an $I \in \Ins^*$. By our choice of $\Ins^*$, we have
$$
\Gamma^{\alpha Y}(I) \geq \left(1-\frac{1}{\lg^3 \lg n}\right) \cdot
\left(1-\frac{1}{Z}\right) \cdot \frac{n \lg n}{\alpha Y}.
$$
i.e. the old data structure has many pairs $(b,i )$ that are
$\alpha Y$-scattered on input $I$. By definition of $\bar{L}^{\alpha
  Y}(I)$, we need to lower
bound the number of such pairs $(b,i)$ that have $\bar{Q}^{\alpha Y}_i(b,I)$
non-empty, i.e. at least one query reporting $b$ in tree-levels
$\{i\alpha Y,\dots,(i+1)\alpha Y-1\}$ is accepted by our new query
algorithm. For this, let $(b,i)$ be a pair that was $\alpha
Y$-scattered for $I$ in the old data structure. By definition of
$\alpha Y$-scattered we know that $Q^{\alpha Y}_i(b,I)$ is non-empty. Now observe that if $Q^{\alpha Y}_i(b,I)$
contains a query that made $0$ probes, then that query is also in
$\bar{Q}^{\alpha Y}(b,I)$. Similarly if $Q^{\alpha Y}_i(b,I)$ contains
a query making its first probe in a $w^3$-popular cell, then that query is also in
$\bar{Q}^{\alpha Y}_i(b,I)$ since we publish all $w^3$-popular
cells. Hence 
$\bar{Q}_i^{\alpha Y}(b,I)$ can be empty only if $Q_i^{\alpha
  Y}(b,I)$ contains no queries making $0$ probes and no queries
probing a $w^3$-popular cell. Since
$(b,i)$ was $\alpha Y$-scattered, this implies $|C_i^{\alpha Y}(b,I)|
\geq \alpha/\lg^6\lg n$. Furthermore, we get that $\bar{Q}_i^{\alpha
  Y}(b,I)$ becomes empty only if none of these cells are in $P^*$.

Letting $\mu = n \lg n/(w \lg^{10}\lg n)$, we
get that $C_i^{\alpha Y}(b,I)$ has a non-zero intersection with the
following fraction of $\mu$-sized cell sets:
\begin{eqnarray*}
1-\frac{\binom{S-|C_i^{\alpha Y}(b,I)|}{\mu}}{\binom{S}{\mu}} \geq
1-\frac{(S-\alpha/\lg^6\lg n)!(S-\mu)!\mu!}{S!(S-\alpha/\lg^6 \lg
  n-\mu)!\mu!} \geq 
1-\frac{(S-\mu)^{\alpha/\lg^6 \lg n}}{(S-\alpha/\lg^6 \lg
  n)^{\alpha/\lg^6 \lg n}} &=&\\
1-\left(\frac{S-\alpha/\lg^6 \lg n+\alpha/\lg^6 \lg n-\mu}{S-\alpha/\lg^6 \lg n}\right)^{\alpha/\lg^6 \lg n} =
1-\left(1-\frac{\mu - \alpha/\lg^6 \lg n}{S-\alpha/\lg^6 \lg
    n}\right)^{\alpha / \lg^6 \lg n}.
\end{eqnarray*}
Since $\alpha=(S w \lg^{18} \lg n)/(n \lg n) = S \lg^8 \lg n / \mu \ll \mu/2$,
this is at least a
\begin{eqnarray*}
1-\left(1-\frac{\mu}{2S}\right)^{\alpha/\lg^6 \lg n} \geq
1-\exp\left(-\alpha \mu /(2 S \lg^6 \lg n)\right) \geq
1-1/\lg n
\end{eqnarray*}
fraction. Since we chose $P^*$ to maximize the number sets
$C_i^{\alpha Y}(b,I)$ having a non-empty intersection, we conclude
that at least
$$
\left(1-\frac{1}{\lg n}\right)\cdot\left(1-\frac{1}{\lg^3 \lg n}\right)
  \cdot \left(1-\frac{1}{Z}\right) \cdot \frac{n \lg n}{\alpha Y}  \geq
\left(1-\frac{1}{Z}-\frac{2}{Z \lg^3 \lg n}\right) \frac{n \lg
  n}{\alpha Y}
$$
sets $\bar{Q}_i^{\alpha Y}(b, I)$ must be non-empty. Since $Z \geq 2$,
we finally conclude
$$
\bar{L}^{\alpha Y}(I) \geq \left(1-\frac{1}{Z}-\frac{1}{\lg^3\lg n}\right)
\frac{n \lg n}{\alpha Y}.
$$

%% file: scattering.tex
\subsection{High Scattering}
\label{sec:highscatter}
In the following, we prove Lemma~\ref{lem:scatter}, i.e. that many
queries have to be scattered. The proof is based on an encoding
argument. Let $\Ins$ be a set of inputs to the ball-inheritance
problem such that $|\Ins|\geq n!/2^n$ and consider a data structure
with $S$ cells of $w$ bits, $p < n \lg n/\lg^9 \lg n$ published bits and
$L^Y(I) \geq (1-1/Z)n \lg n/Y$ for all $I \in \Ins$, for some
parameters $Z \geq 2$ and $64 \lg w \leq Y \leq \lg n/\alpha$ where
$\alpha=\alphaval$. Assume for contradiction that the data structure also
satisfies:
\begin{eqnarray}
\Gamma^{\alpha Y}(I) < \left(1-\frac{1}{\lg^3 \lg n}\right) \cdot
\left(1-\frac{1}{Z}\right) \cdot \frac{n\lg n}{\alpha Y}
\end{eqnarray} for more than
$|\Ins|/2$ of the inputs $I \in \Ins$. We call such inputs
\emph{interesting}. We show that all interesting inputs can be
uniquely encoded (and decoded) into a string of less than $\lg(n!)-n-1
\leq \lg(|\Ins|)-1$ bits. This is clearly a contradiction since there
are more than $|\Ins|/2$ interesting inputs. For the remainder of the
section, we implicitly work with this contradictory data structure,
e.g. whenever we say an interesting input, we mean an interesting
input for the contradictory data structure satisfying all of the
above.

The encoding we present below exploits that an interesting input must
have many leaf indices $b \in [n]$ that are not $\alpha Y$-scattered
and at the same time, the $Y$-level-support of $b$ is high. These two
properties combined implies that such a leaf $b$ is reported by many
queries that read within a small set of non-popular cells. Thus the
data structure has in some sense managed to \emph{route} queries
reporting the same leaf to the same memory cell. This should not be
possible if the queries are sufficiently far apart in the
ball-inheritance tree, at least not without a large number of
published bits, see the intuition in Section~\ref{sec:ideas}. Turning this into a concrete property we can use in an
encoding argument requires a few definitions.

\begin{definition}
Let $I$ be an interesting input and let $(q_1,q_2)$ be a pair of
queries. We say that $(q_1,q_2)$ is a \emph{ball-edge} if $q_1$ and
$q_2$ report the same leaf index on input $I$, and furthermore, $q_1$
is at a higher level in the ball-inheritance tree than $q_2$. The
\emph{length} of a ball-edge $(q_1,q_2)$ is the number of levels
between $q_1$ and $q_2$ in the ball-inheritance tree. Ball-edges of length $1$ are called \emph{regular} edges and ball-edges longer than $1$ are called \emph{shortcut} edges.
\end{definition}

\begin{observation}
\label{obs:paths}
Let $I$ be an interesting input. Then $I$ is uniquely determined from
the set of all the $n \lg n$ regular ball-edges.
\end{observation}

\begin{proof}
From the set of all regular edges $(q_{1,1},q_{1,2}),\dots,(q_{n \lg
  n, 1}, q_{n \lg n, 2})$ we find all pairs $(q_{i,1},q_{i,2}),
(q_{j,1},q_{j,2})$ such that $q_{i,2}=q_{j,1}$. Collectively, all
these pairs form $n$ paths $P_0,\dots,P_{n-1}$, each of length $\lg
n$. All the queries in a path $P_i$ must necessarily report the same
leaf index, and no query in a path $P_j$, where $j \neq i$, reports the
same leaf index. Each path $P_i$ has the form $(q_0,
q_1),(q_1,q_2),\dots,(q_{\lg n-1},q_{\lg n})$, where $q_i$ is a query
at the $i$'th level of the ball-inheritance tree. Now recall that a
query $q$ is specified by a node in the ball-inheritance tree and an
index (rank) into that node's list of balls. Thus the query $q_{\lg
  n}$ tells us the leaf index $b$ returned by all queries in $P_i$ on
input $I$. The query $q_0$ tells us the rank of the ball reaching $b$
in the root's ball list (the rank amongst all balls). This information
for every $b$ specifies $I$ completely.
\end{proof}

With Observation~\ref{obs:paths} in mind, we set out to give a
succinct encoding of all regular ball-edges. The trick is to encode a
set of shortcut edges cheaply and use the information they provide to
avoid explicitly encoding the regular edges spanning the same subset
of tree levels. For the shortcut edges to collectively save many bits,
we need them to be non-overlapping in the following sense:

\begin{definition}
Let $(q_1,q_2)$ and $(q_3,q_4)$ be two ball-edges and let
$\ell_1,\ell_2,\ell_3$ and $\ell_4$ denote the tree levels where
queries $q_1,q_2,q_3$ and $q_4$ are asked respectively. Then the two
edges are \emph{non-overlapping} if either the queries return different leaf
indices, or if the two sets of level indices
$\{\ell_1,\ell_1+1,\dots,\ell_2\}$ and
$\{\ell_3,\ell_3+1,\dots,\ell_4\}$ spanned by the edges have an intersection of size at most
$1$. Otherwise, they are \emph{overlapping}.
\end{definition}

We are finally ready to state the main lemma allowing us to compress
interesting inputs during our encoding steps:
\begin{lemma}
\label{lem:nicepaths}
Let $I$ be an interesting input. Then there exists a set of shortcut
edges $$P=\{(q_{1,1},q_{1,2}),\dots,(q_{m,1},q_{m,2})\}$$ of lengths
$\ell_1,\dots,\ell_m$ satisfying the following:
\begin{enumerate}
\item All edges in $P$ are non-overlapping.
\item $\ell_i \geq 64 \lg w$ for all $i$. 
\item $\sum_i \ell_i = \Omega(n \lg n / \lg^8 \lg n)$.
\item For all $i$, the queries $q_{i,1}$ and $q_{i,2}$ make their
  first probe in the same cell on input $I$, and that cell is not
  $w^3$-popular. 
\end{enumerate}
\end{lemma}
Note that this in particular implies that all the queries promised by
\ref{lem:nicepaths} are accepted for input $I$ since they all make at
least $1$ probe.
We defer the proof of Lemma~\ref{lem:nicepaths} to
Section~\ref{sec:findingpaths} and instead move on to show how we use it in our
encoding and decoding procedures to obtain an encoding of each
interesting input in less than $\lg(n!)-n-1$ bits. As remarked
earlier, we encode the set of all regular edges, which by
Observation~\ref{obs:paths} allows us to recover the interesting
input. There are two main ideas to have in mind: First, we will use a
shortcut edge of length $\ell_i$ to avoid explicitly encoding the
$\ell_i$ overlapping regular edges. Assuming a saving of one bit per
regular edge, we save a total of $\sum_i \ell_i$ bits. Secondly, each
edge $(q_{i,1},q_{i,2}) \in P$ consists of two queries probing the
same cell, and that cell is not $w^3$-popular. Since less than $w^3$
queries probe that cell, the edge can be encoded in $6 \lg w$ bits by
specifying it as a pair amongst the queries probing the cell. Thus
encoding a shortcut edge saves us $\ell_i - 6\lg w = \Omega(\ell_i)$
bits. Summed over all shortcut edges gives us a saving of
$\Omega(\sum_i \ell_i)$ bits in total. This saving happens precisely
because the data structure was able to route distant queries reporting
the same leaf index to the same memory cells, and that memory cell is
read by only few queries.
%

\paragraph{Encoding.}
In this paragraph, we show how we encode the regular edges for a given
interesting input $I$. The encoding uses the set of shortcut edges $P$
promised by Lemma~\ref{lem:nicepaths}.

From $I$, define for every memory cell $c$ (which is also an index in
$[S]$), the set $V_c$ consisting of all (accepted) queries making
their first probe in $c$ on input $I$. Also let $W_c \subset V_c
\times V_c$ denote the set of shortcut edges $(q_1,q_2) \in P$ such
that $q_1,q_2 \in V_c$. Note that $W_c$ is non-empty only for cells
$c$ that are not $w^3$-popular and every shortcut edge $(q_1,q_2) \in
P$ is contained in precisely one set $W_c$.

With these definitions, our encoding procedure is as follows:
\begin{enumerate}
\item First we write down the published bits on input $I$. This costs
  no more than $n \lg n/\lg^9 \lg n + O(\lg n)$ bits (the $O(\lg n)$
  bits specify the number of published bits).
\item Next we write down a bit vector $v$ with $S$ bits, one for each
  memory cell. The $c$'th bit is $1$ iff $W_c$ is non-empty (which
  also implies that $c$ is not $w^3$-popular).
\item Now for $c=0,\dots,S-1$ in turn, we check whether the $c$'th bit
  of $v$ is $1$. If so, we first write down $|W_c|$. Note that if the
  $c$'th bit is $1$, then $c$ is not $w^3$-popular. Therefore $|W_c|
  \leq |V_c|^2 \leq w^{6}$ and the counter takes only $6 \lg w$
  bits. After having written down the count, we write down each of the
  shortcut edges $W_c$. Each such shortcut edge is specified using $2 \lg
  |V_c| \leq 6 \lg w$ bits by writing down the corresponding pair of
  queries in $V_c$.
\item The final step encodes a subset of the regular edges. This is
  done by recursively visiting the nodes of the ball-inheritance tree,
  starting at the root. For a node $v$ at depth $d$, the encoding
  procedure does as follows: Let $q_1,\dots,q_{n/2^d}$ denote the
  sorted list of queries at $v$, i.e. $q_i$ asks for the leaf index
  reached by the $i$'th ball in $v$'s ball list. Go through the
  queries in turn from $i=1,\dots,n/2^d$. For each $q_i$, let
  $(q_i,\dest(q_i))$ denote the regular edge having $q_i$ as origin,
  i.e. $\dest(q_i)$ gives the query at depth $d+1$ returning the same
  leaf index as $q_i$ on input $I$. We now check whether there are any
  shortcut edges in $P$ that overlap with $(q_i,\dest(q_i))$. If not, we
  append one bit to the encoding, specifying whether $\dest(q_i)$ is a query
  in the left child or the right child of $v$. Otherwise (there is an
  overlapping shortcut edge), we do not write any bits for $q_i$. We
  then recurse on the two children of $v$ (first the left, then the
  right), using their respective sorted lists of queries. The
  recursion ends when reaching the leaves.
\end{enumerate}

\paragraph{Decoding.}
In the following, we show how we recover all the regular edges from
the encoding above. By Observation~\ref{obs:paths} this also recovers
$I$. The decoding procedure is as follows:
\begin{enumerate}
\item First we read the published bits from the encoding. From the
  published bits alone, we determine for every query whether it is
  accepted or not, and which cell it probes first in case it is
  accepted. From this information, we can reconstruct the sets $V_c$
  for all $c$.
\item Next we read the bit vector $v$ specified in step 2. of the
  encoding procedure. This tells us which cells $c$ that are not
  $w^3$-popular and where $W_c$ is non-empty. For each such cell $c$
  in turn, we recover $W_c$ from the bits written in step 3. of the
  encoding procedure. Since $\cup_c W_c = P$, we have also recovered $P$.
\item Our last decoding step recovers the regular edges. We
  do this recursively, starting at the root node:

  For a node $v$ at depth $d$ in the ball-inheritance tree (starting
  at the root), let $q_1,\dots,q_{n/2^d}$ denote the sorted list of
  queries at $v$, i.e. $q_i$ asks for the leaf reached by the $i$'th
  ball in $v$'s list of balls. Now for $i=1,\dots, n/2^d$ in turn,
  consider the query $q_i$ and assume we have already recovered all
  regular edges having an origin in an ancestor of $v$ and all regular
  edges having a query $q_{i'}$ as origin, where $i'<i$.
 
  From the already recovered regular edges, we determine the sequence of
  regular edges 
$$A(q_i) = (p_0,p_1),(p_1,p_2),\dots,(p_{d-1}, q_i)$$ corresponding to the same
  ball as $q_i$ on input $I$ ($p_{d'}$ is the query at level $d' < d$
  returning the same leaf index as $q_i$ on input $I$). Observe that
  we can determine this sequence since the origin of each edge is the
  destination of the preceding edge. 

  We also count the number of indices $i' < i$ such that $q_{i'}$
  returns a leaf in the left subtree of $v$ (this can be seen directly
  from the regular edges already recovered for all $q_{i'}$ with $i' <
  i$). Denote this number of indices by $R$.

  Next we determine whether $q_i$ returns a leaf index in the left or
  right subtree. This is done as follows: First we check if there is a
  shortcut edge $(p,r) \in P$ that overlaps with the edge $(q_i,\dest(q_i))$
  (we still do not know $\dest(q_i)$). This is done by examining each query
  $p$ in the edges of $A(q_i)$ (including $q_i$ itself) and checking
  if there is shortcut edge $(p,r) \in P$ having $p$ as origin. If so,
  we check the depth of $r$ (since $r$ is a query, and queries are
  specified by a node in the ball-inheritance tree, this information
  is available). If the depth is greater than $d$, we conclude that
  $(p,r)$ overlaps with $(q_i,\dest(q_i))$. Otherwise it does not. If there
  was an overlapping shortcut edge $(p,r) \in P$, $r$ tells us the
  subtree containing the leaf returned by $q_i$. Otherwise, we read
  off the next bit of the part of the encoding written in step 4. of
  the encoding procedure. This bit tells us the subtree containing the
  leaf index returned by $q_i$. Let $w$ denote the child of $v$ whose
  subtree contains the leaf returned by $q_i$. If $w$ is a left child,
  let $\Delta=R$ denote the number of indices $i' < i$ such that $q_{i'}$
  also returns a leaf index in $w$'s subtree. If $w$ is a right child,
  we have $\Delta=i-1-R$.

  Now observe that since the ordering of balls remains the same in
  each sublist, $q'_\Delta=\dest(q_i)$, where $q'_\Delta$ is the $\Delta$'th
  query in $w$'s list of queries. Hence we have recovered the regular
  edge $(q_i, \dest(q_i)) = (q_i, q'_\Delta)$.

  After processing all $i=1,\dots,n/2^d$, we recurse to the two
  children of $v$ (first the left and then the right). When the entire
  process terminates, we have recovered all the regular edges and
  hence $I$.
\end{enumerate}

\paragraph{Analysis.}
What remains is to analyze the size of the encoding and derive a
contradiction. Step 1. of the encoding procedure costs at most
$n \lg n/\lg^9\lg n+O(\lg n)$ bits. Step 2. costs $S$ bits. For step 3.,
observe that the $6 \lg w$ bit counter for $|W_c|$ can be charged to
at least one shortcut edge in $P$. Similarly, each shortcut edge
specified in step 3. also costs at most $6 \lg w$ bits. But each
shortcut edge has length at least $64\lg w$ and hence the total number
of bits written down in step 3. is bounded by $\sum_i \ell_i \cdot
(12/64) = (3/16) \cdot \sum_i \ell_i$. For step 4., note that a
regular edge only adds a bit to the encoding size if it is not
overlapping with any of the shortcut edges in $P$. But the shortcut
edges in $P$ are non-overlapping themselves and the $i$'th such edge
overlaps with $\ell_i$ regular edges. Thus step 4. costs at most $n
\lg n - \sum_i \ell_i$ bits. Summarizing, the encoding uses:
$
n\lg n + n \lg n/\lg^9 \lg n + S - \Omega\left(\sum_i \ell_i\right)
$
bits. 
Since $w = \Omega(\lg n)$, and $\alpha=\alphaval \leq \lg n$ implies  $S = O(n \lg n/ \lg^{18} \lg
n) $, we conclude that our encoding
uses
$
n \lg n - \Omega(n \lg n/\lg^8 \lg n)
$ 
bits. But $n \lg n \leq \lg(n!)+\Theta(n)$ and thus we have arrived
at a contradiction since our encoding uses $\lg(n!)-\omega(n)$ bits.

%% file: shortcuts.tex
\subsection{Finding Shortcut Edges}
\label{sec:findingpaths}
This section is devoted to the proof of Lemma~\ref{lem:nicepaths}. Recall that an interesting input $I$ satisfies:
\begin{enumerate}
\item $L^Y(I) \geq (1-1/Z)n \lg n/Y.$
\item $\Gamma^{\alpha Y}(I) < \left(1-\frac{1}{\lg^3 \lg n}\right) \cdot
\left(1-\frac{1}{Z}\right) \cdot \frac{n\lg n}{\alpha Y}.$
\end{enumerate}
for some $64 \lg w \leq Y \leq \lg n/\alpha$, where $\alpha=\alphaval$. The
first step in finding the set of regular edges $P$ claimed by
Lemma~\ref{lem:nicepaths} is to show that there must be many leaf
indices $b$ that are not $\alpha Y$-scattered, and at the same time has
high $Y$-level support (there is a query reporting it in many levels
of the tree). Formally, we show:

\begin{lemma}
\label{lem:compress}
Let $I$ be an interesting input. Then there is at least $\frac{n \lg
  n}{\alpha Y \lg^4 \lg n }$ pairs $(b, i) \in [n] \times [\lg n/(\alpha Y)]$, such that:
\begin{enumerate}
\item $(b,i)$ is not $\alpha Y$-scattered.
\item There are at least $\alpha/\lg^4 \lg n$ indices $j \in [\alpha]$ such that $Q^Y_{i \alpha+j}(b, I)$ is non-empty.
\end{enumerate}
\end{lemma}

\begin{proof}
Let $(\ball, \randi)$ be uniform random in $[n] \times [\lg n/(\alpha
  Y)]$ and let $\Count$ be the random variable giving the number of
indices $j \in [\alpha]$ such that $Q^Y_{\randi \alpha + j}(\ball, I)$
is \emph{empty}. By linearity of expectation, we have $\E[\Count] \leq
\alpha/Z$. Furthermore, $\Count$ is non-negative, thus we may use
Markov's inequality to conclude:
$$\Pr\left[\Count > \left(1-\frac{1}{\lg^4 \lg n}\right)\cdot \alpha\right] < \frac{1}{(1-1/\lg^4 \lg n)Z}.$$ 

At the same time, we have 
\begin{eqnarray*}
\Pr[(\ball, \randi) \textrm{ is not }
  \alpha Y\textrm{-scattered}] &>& 1-\left(1-\frac{1}{\lg^3 \lg n}\right)\left(1-\frac{1}{Z}\right) \\
&=& \frac{1}{Z}+\frac{1}{\lg^3 \lg n}-\frac{1}{Z \lg^3 \lg n}.
\end{eqnarray*}

From a union bound, it follows that
\begin{eqnarray*}
\Pr\left[\Count \leq  \left(1-\frac{1}{\lg^4 \lg n}\right)\cdot \alpha \bigwedge (\ball, \randi) \textrm{ is not }
  \alpha Y\textrm{-scattered}\right] &\geq& \\
1-\frac{1}{(1-1/\lg^4 \lg n)Z}-\left(1-\frac{1}{Z}-\frac{1}{\lg^3 \lg n}+\frac{1}{Z\lg^3 \lg n}\right) &=&\\
\left(1-\frac{1}{1-1/\lg^4 \lg n} - \frac{1}{\lg^3 \lg n}\right) \cdot \frac{1}{Z} + \frac{1}{\lg^3 \lg n} &=&\\
\frac{1}{\lg^3 \lg n} - \left(\frac{1}{\lg^4 \lg n(1-1/\lg^4 \lg n)}+\frac{1}{\lg^3 \lg n}\right)\cdot \frac{1}{Z} &\geq&\\
\frac{1}{\lg^4 \lg n}.
\end{eqnarray*}
Here the last inequality follows from $Z \geq 2$. Since $(\ball, \randi)$ is uniform random, the lemma follows.
\end{proof}

We call the pairs $(b, i) \in [n] \times [\lg n/(\alpha Y)]$ that
satisfy the properties in Lemma~\ref{lem:compress}
\emph{compressable pairs}. For each compressable pair we define the
\emph{representative query set}, denoted $\Rep^{\alpha Y}_{i}(b, I)$, as the set
consisting of one query from each non-empty set $Q^Y_{i\alpha + j}(b,
I)$ where $j \in [\alpha]$ (the choice of query from $Q^Y_{i\alpha +
  j}(b,I)$ is irrelevant).  Define also the set of cells
$\RepCells^{\alpha Y}_i(b,I)$ consisting of the first cell probed by each query
in $\Rep^{\alpha Y}_i(b,I)$ on input $I$. The representative query sets have the
following properties:
\begin{lemma}
\label{lem:ready}
Let $(b, i) \in [n] \times [\lg n/(\alpha Y)]$ be a compressable pair for an interesting input $I$. Then
\begin{enumerate}
 \item  $|\Rep^{\alpha Y}_i(b,I)| \geq \alpha/\lg^4 \lg n$.
\item $|\RepCells^{\alpha Y}_i(b,I)|\leq \alpha/\lg^6 \lg n$.
\item There are no queries in $\Rep^{\alpha Y}_i(b,I)$ that makes 0 probes on input $I$.
\item $\RepCells^{\alpha Y}_i(b,I)$ does not contain a $w^3$-popular cell.
\end{enumerate}
\end{lemma}
\begin{proof}
Property 1. follows from property 2. in Lemma~\ref{lem:compress}. Since $\Rep^{\alpha Y}_i(b,I) \subseteq Q^{\alpha
  Y}_i(b, I)$ it follows that $\RepCells^{\alpha Y}_i(b,I) \subseteq C_i^{\alpha
  Y}(b,I)$ and thus from property 1. in Lemma~\ref{lem:compress} and
property 3. in Definition~\ref{def:scatter} we conclude
$|\RepCells^{\alpha Y}_i(b,I)|\leq \alpha/\lg^6 \lg n$. Property 3. and 4. follows from the same argument.
\end{proof}

Lemma~\ref{lem:ready} sets the stage for finding the shortcut edges
$P$. Examining the lemma, we see that properties 1. and 2. together
imply that there must be many queries in levels $[i\alpha Y :
  (i+1)\alpha Y-1]$ that report $b$ and at the same time probe a small
set of cells (several of them must probe the same cell). In addition,
the probed memory cells are not $w^3$-popular as required by
Lemma~\ref{lem:nicepaths}.

Now consider a compressable pair $(b,i) \in [n] \times [\lg n/(\alpha Y)]$
for an interesting input $I$. Assume 
$$(q_{1,1},q_{1,2}),\dots,(q_{m,1},q_{m,2})$$
is a set of shortcut edges where each $(q_{j,1},q_{j,2})$ consists of a
pair of queries in $\Rep^{\alpha Y}_i(b,I)$. Then that shortcut edge is
non-overlapping with any shortcut edge defined from queries in another
set $\Rep^{\alpha Y}_{i'}(b', I)$ where at least one of $i'$ and $b'$
is different from $i$ and $b$ respectively. To see this, observe that
if $b'\neq b$, then the edges are non-overlapping since the
corresponding queries report different leaf indices. If $i' \neq i$,
then the set of levels spanned by the edges are disjoint since any
query in $\Rep^{\alpha Y}_i(b,I)$ must be in levels $[i \alpha Y :
  (i+1)\alpha Y -1]$. With this insight, we set out to construct
non-overlapping edges for each $\Rep^{\alpha Y}_i(b,I)$ where $(b,i)$
is a compressable pair. Taking the union of the constructed sets
leaves us with a set of shortcut edges that is still non-overlapping.

\paragraph{Finding Shortcut Edges for a Compressable Pair.}
Let $(b,i) \in [n] \times [\lg n/(\alpha Y)]$ be a compressable
pair. We collect shortcut edges $(q_1,q_2)$ such that both $q_1,q_2
\in \Rep^{\alpha Y}_i(b,I)$ and $q_1$ and $q_2$ make their first probe
in the same cell in $\RepCells^{\alpha Y}_i(b, I)$ on input $I$. While
finding these shortcut edges, we ensure they are non-overlapping and
that the sum of their lengths is large.

Our procedure for constructing a set of shortcut edges $E$ is as
follows: Let $m = |\Rep^{\alpha Y}_i(b,I)|$ and order the queries in
$\Rep^{\alpha Y}_i(b,I)$ by depth in the ball-inheritance tree. Let
$q_1,\dots,q_m$ denote the resulting sequence of queries where $q_1$
is at the lowest depth (closest to the root). Now initialize $j \gets
1$ and iteratively consider query $q_j$: When examining query $q_j$,
let $c \in \RepCells^{\alpha Y}_i(b,I)$ be the first cell probed by
$q_j$ and let $V_c \subseteq \Rep^{\alpha Y}_i(b,I)$ be the subset of
queries in $\Rep^{\alpha Y}_i(b,I)$ that also make their first probe
in $c$. If $q_j$ is amongst the deepest $2$ queries in $V_c$, we
simply continue by setting $j \gets j+1$. Otherwise, let $q_h$ be the
deepest query in $V_c$. We add the shortcut edge $(q_j,q_h)$ to $E$
and update $j \gets h$. This procedure terminates when $j=m$.

\begin{lemma}
The above procedure outputs a set of shortcut edges $E$ such that:
\begin{enumerate}
\item The
edges in $E$ are non-overlapping.
\item The sum of their lengths is $\Omega(\alpha Y/\lg^4 \lg n)$.
\item Each edge has length at least $Y \geq 64 \lg w$.
\item For each edge $(q_1,q_2) \in E$, the queries $q_1$ and $q_2$
  make their first probe in the same cell, and that cell is not
  $w^3$-popular.
\end{enumerate}
\end{lemma}
\begin{proof}
Property 1. and 4. follows trivially from the construction
algorithm. Property 3. follows since each added edge is amongst a pair
of queries with at least one query from $\Rep^{\alpha Y}_i(b,I)$ in
between. But the queries in $\Rep^{\alpha Y}_i(b,I)$ all appear in
distinct $Q^Y_{i'}(b,I)$ and hence each edge in $E$ must have length
at least $Y \geq 64 \lg w$.

For property 2., define for each edge $e=(q_{j_1},q_{j_2}) \in E$ the set of
queries $Q_e \subseteq \Rep^{\alpha Y}_i(b,I)$ for which each $q \in
Q_e$ appears in a tree level in between the levels of $q_{j_1}$ and
$q_{j_2}$. By the arguments above, the length of $e$ must be at least
$Y|Q_e|$. Thus we bound $\sum_{e \in E}|Q_e|$. For this, let
$j_1,j_2,\dots,j_k$ be the distinct values taken by the variable $j$
in the construction algorithm above. We have $\sum_{e \in E}|Q_e| =
\sum_{h=2}^k (j_h-j_{j-1}-1)$. This is bounded by
$\left(\sum_{h=2}^k(j_h-j_{h-1})\right)-(k-1) \geq j_k - j_1 - k +1 =
m-k$. What remains is to bound $k$. Observe that for each cell $c
\in \RepCells_i^{\alpha Y}(b,I)$, there are only two of the queries in
$V_c$ that can cause $j$ to be incremented by less than $2$. Therefore
we must have $k \leq 2|\RepCells^{\alpha Y}_i(b,I)| + m/2$. But $m =
|\Rep^{\alpha Y}_i(b,I)| \geq \alpha/\lg^4 \lg n$ and
$|\RepCells^{\alpha Y}_i(b,I)| \leq \alpha/\lg^6 \lg n$ by Lemma~\ref{lem:ready}. Hence we
conclude that the sum of the lengths of edges in $E$ is lower bounded
by $\Omega(\alpha Y/\lg^4\lg n)$.
\end{proof}

By Lemma~\ref{lem:compress}, we have at least $n\lg n/(\alpha Y \lg^4 \lg n)$
compressable pairs. Taking the union of the edge sets constructed for
each such pair completes the proof of Lemma~\ref{lem:nicepaths}.